\documentclass[11pt,twoside]{article} 
\usepackage{acta-info}
\usepackage{euler}

\newcommand{\vol}{\textbf{3,} 2 (2011) 205--223}   
\newcommand{\amss}{\amssubj{%
    11-04, 
    11Y11, 
    68W99 
  }} 
\newcommand{\CRs}{\CRsubj{%
    F.2.1, 
    E.1
}}   
\newcommand{\keyw}{\keywords{%
    sieve, cache memory, number theory, primality
}}

\begin{document}

\title{%
  Cache optimized linear sieve }
\maketitle

\twoauthors{%
  \href{http://compalg.inf.elte.hu/~ajarai}{Antal J\'ARAI} }{%
  \href{http://www.elte.hu}{E\"otv\"os Lor\'and University} }{%
  \href{mailto:ajarai@moon.inf.elte.hu}{ajarai@moon.inf.elte.hu} }{%
  \href{http://compalg.inf.elte.hu/~vatai}{Emil VATAI} }{%
  \href{http://www.elte.hu}{E\"otv\"os Lor\'nd University} }{%
  \href{mailto:emil.vatai@gmail.com}{emil.vatai@gmail.com} 
}


\short{%
  A. J\'arai, E. Vatai
}{%
  Cache optimized linear sieve
}

\begin{abstract}
  Sieving is essential in different number theoretical algorithms.
  Sieving with large primes violates locality of memory access, thus
  degrading performance. Our suggestion on how to tackle this problem
  is to use cyclic data structures in combination with in-place
  bucket-sort.

  We present our results on the implementation of the sieve of
  Eratosthenes, using these ideas, which show that this approach is
  more robust and less affected by slow memory.
\end{abstract}



\section{Introduction}
\label{sec:introduction}

In this paper we present the results obtained by implementing the
sieve of Eratosthenes \cite{bressoud} using the methods described at
the 8th Joint Conference on Mathematics and Computer Science in Kom\'arno \cite{macs}. In
the first section, the problem which is to be solved by the algorithm
and some basic ideas about implementation and representation are
presented.  In Section \ref{sec:addr-memory-local}, the methods to speed up the execution are
discussed. In Section \ref{sec:speed-up-algo} the numerical data of the measurement of
run\-times is provided on two different platforms in comparison with
the data found at \cite{tos}.


Given an array (of certain size) and a set $P$ of $(p,q)$ pairs, where
$p$ is (usually) a prime and $0\le q<p$ is the offset (integer)
associated with $p$. A sieving algorithm for each $(p,q)\in P$ pair
performs an action on every element of the array with a valid index
$i=q+mp$ (for $m\ge0$ integers).

Sieving with small primes can be considered a simple and efficient
algorithm: start at $q$ and perform the action for sieving, then
increase $q$ by $p$ and repeat. But when sieving with large primes,
larger than the cache, memory hierarchy comes into play. With these
large primes, sieving is not sequential (i. e. $q$ skips great
portions of memory), thus access to the sieve array is not sequential
and this causes the program to spend most of its time waiting to
access memory, because of cache misses.

\subsection{Sieve of Eratosthenes}
\label{sec:sieve-eratosthenes}

The sieve of Eratosthenes is the oldest algorithm for generating
consecutive primes. It can be used for generating primes ``by hand''
but it is also the simplest and most efficient way to generate
consecutive primes of high vicinity using computers. The algorithm is
quite simple and  well-known. Starting with the number $2$,
declare it as prime and mark every even number as a composite
number. The first number, which is not marked is $3$. It is declared
as the next prime, so all the numbers divisible by $3$ (that is every
third number) is sieved out (marked), and so on.

\subsection{Basic ideas about implementation}
\label{sec:basic-ideas-about}

The program finds all primes in an interval $[u,v]\subset\mathbb{N}$
represented in a bit table. The program addresses the issue of memory
locality by sieving the $[u,v]$ in sub\-intervals of predefined size,
called segments, which can fit in the cache, thus every segment of the
sieve table needs to pass through the cache only once. For simplicity
and efficiency the size of segment is presumed to be the power of
two.

\bigskip

Because every even number, except $2$, is a composite, a trivial
improvement is to represent only the odd numbers, and cutting the size
of the task at hand in half.

\subsubsection{Input and output}
\label{sec:parameters}

The program takes three input parameters: the base 2 logarithm of the
segment size denoted by $l$,\footnote{cache size $\approx$ segment
  size $=2^l$ bits $=2^{l-3}$ bytes $=2^{l+1}$ numbers represented,
  because only odd numbers are represented in the bit table.} and
instead of the explicit interval boundaries $u$ and $v$, the ``index''
of the first segment denoted by $f$, and the number of segments to be
sieved denoted by $n$, is given. This means, the numbers between
$u=f2^{l+1}+1$ and $v=(f+n)2^{l+1}$ are sieved. For simpler comparison
with results from \cite{tos}, the exponent of the approximate midpoint
of the interval can also be given as an input parameter instead of
$f$\footnote{Of course, this just relieves the user from the tedious
  task of calculating $f$ by hand for the given exponent of the
  midpoint of the interval, but internally the flow of the program was
  same as if $f$ is given.}.

As for the output, the program stores the finished bit table in a file
in the \texttt{/tmp} directory, with the parameters written in the
file\-name.

\subsection{Sieve table and segments}
\label{sec:sieve-table-segments}


\begin{definition}[Sieve table, Segments]
  The sieve table $S$ (for the above gi\-ven parameters) is an array of
  $n2^l$ bits. For $0\le j < n2^l$, the bit $S_j$ represents the odd
  number $2(f2^l+j)+1 \in [u,v]=[f2^{l^+1},(f+n)2^{l+1}]$.  $S_j$ is
  initialized to $0$ (which indicates that $2(f2^l+j)+1$ is prime);
  $S_j$ is set to $1$, if $2(f2^l+j)+1$ is sieved out i.e.\ it is
  composite.

  The $t$-th segment, denoted by $S^{(t)}$ is sub\-table of the $t$-th
  $2^l$ bits of the sieve table $S$, i.e.\ $S^{(t)}_q= S_{t2^l+q}$ for
  $0\le q < 2^l$ and $0\le t <n$.
  
\end{definition}

After $p$ sieves at $S_j$, marking $2(f2^l+j)+1$ as a composite, the
next odd composite divisible by $p$ is the $2(f2^l+j)+1+2p=
2(f2^l+j+p)+1 $, so the index $j$ has to be incremented only by $p$.
That is, not representing the even numbers doesn't change the sieving
algorithm, except the calculation of the offsets (described in Lemma
\ref{lem:q}).

\subsection{Initialization phase}
\label{sec:initialization-phase}

For every prime $p$, the first composite number not marked by smaller
primes, will be $p^2$, i.e.\ sieving with $p$ can start from $p^2$. To
sieve out the primes in the $[u,v]$ interval, only the primes
$p\le\sqrt{v}$ are needed. Finding these primes and calculating the
$q$ offsets, so that $q\ge 0$ is the smallest integer satisfying $p
\mid 2(fs^s+q)+1$ is the initialization phase.  Presumably $\sqrt{v}$
is small ($\sqrt{v}<u$), and finding primes $p<\sqrt{v}$ (and
calculating offsets) can be done quickly.

\begin{definition}
  The set of primes, found during the initialization of the sieve is
  called the base.  $P=\{ p \text{ prime} : 2 < p \le\sqrt{v} \}$
\end{definition}

\section{Addressing memory locality}
\label{sec:addr-memory-local}

Because the larger the prime, the more it violates locality of memory
access when sieving, the basic idea is to treat primes of different
sizes in a different way, and process the sieve table by segments in a
linear fashion, loading each segment in the cache only once and
sieving out all the composites in it.

\subsection{Medium primes}

The primes $p<2^l$ are \emph{medium primes}.  Segment-wise sieving
with medium primes is simple: $(p,q)$ prime-offset pairs with
$p\le2^l$ and $0\le q<p$ are stored.  Each prime marks at least one
bit in each segment.  For each prime $p$, starting from $q$, every
$p$-th bit has to be set, by sieving at the offset $q$ and then
incrementing it to $q\gets q+p$ while $q<2^l$.  Now $q$ would sieve in
the next segment, so the offset is replaced with $q-2^l$. The first
offset for a prime $p$ and the given parameters $f$ and $l$ can be
found using the following Lemma.
%

\begin{lemma}
  \label{lem:q}
  For each odd prime $p$ and positive integers $l$ and $f$, there is a
  unique offset $0\le q < p$ satisfying:
  \begin{equation}
    \label{eq:lem}
    p\mid 2(f2^l+q)+1
  \end{equation}
\end{lemma}
\begin{proof}
  Rearranging \eqref{eq:lem} gives $f2^{l+1}+2q+1=mp$ for some $m$.
  The integer $m$ has to be odd, because the left hand side and $p$
  are odd.  The equation can further be rearranged to a form, which
  yields a coefficient and something similar to a reminder:
  \begin{equation*}
    f2^{l+1}=(m-1)p+(p-(2q+1)).
  \end{equation*}
  The last term is even, so if the remainder $r=f2^{l+1} \mod p$ is
  even, then $q=(p-r-1)/2$ satisfies $0\le q < p$ and
  \eqref{eq:lem}. If $r$ is odd, then $q=(2p-r-1)/2$ satisfies $0\le q
  < p$ and \eqref{eq:lem}. Because $r$ is unique, $q$ is also unique.
\end{proof}

\subsection{Large primes}
\label{sec:large-primes}

The primes $p>2^l$ are \emph{large primes}. These primes ``skip''
segments i.e.\ if a bit is marked in one segment by the large prime
$p$, (usually) no bit is marked by the prime $p$ in the adjacent
segment.  The efficient administration of large primes is based on the
following observation:

\begin{lemma}\label{lem:kkpo}
  If the prime $p$, which satisfies the condition $k2^l\le p <
  (k+1)2^l$ (for some integer $k\ge0$), marks a bit in $S^{(t)}$, then
  the next segment where the sieve marks a bit (with $p$) is the
  segment $S^{(t')}$ for $t'=t+k$ or $t'=t+k+1$.
\end{lemma}


\begin{proof}
  If the prime $p$ marks a bit in $S^{(t)}$ then it is the
  $S^{(t)}_q$ bit, for some offset $0\le q < 2^l$. The $S_{t2^l+q}$
  bit is marked first, then the $S_{t2^l+q+p}$ bit, so the index of
  the next segment is $t'=\left\lfloor (t2^l+q+p)/2^l \right \rfloor$, thus 
  \begin{equation*}
    t+k 
    = 
    \frac{t2^l+0+k2^l}{2^l} 
    \le  
    \underbrace{\left\lfloor \frac{t2^l+q+p}{2^l} \right \rfloor}_{=t'}
    < 
    \frac{t2^l+2^l+(k+1)2^l}{2^l} 
    = 
    t+k+2.
  \end{equation*}
\end{proof}

\subsection{Circles and buckets}
\label{sec:circles-buckets}

The goal is, always to have the right primes available for sieving at
the right time.  This is done by grouping primes of the same magnitude
together in so called circles, and within these circles grouping them
together by magnitude of their offsets in so called buckets.

\begin{definition}[Circles and Buckets] 
  \label{def:cb} 
  A circle (of order $k$, in the $t$-th state) denoted by $C^{k,t}$ is
  sequence of $k+1$ buckets $B^{k,t}_d$ (where $0\le d \le k$).  Each
  bucket contains exactly those $(p,q)$ prime-offset pairs, which have
  the following properties:
  \begin{align}
    &k2^l<p<(k+1)2^l           \label{inv1}\\
    &0\le q <\max\{p,2^l\}     \label{inv2}\\
    &p\mid 2\bigl((f+t+d-b+k+1)2^l+q\bigr)+1
    &\textrm{ if }  0\le d<b    \label{inv3}\\
    &p\mid 2\bigl((f+t+d-b)2^l+q\bigr)+1 
    &\textrm{ if }b\le d\le k   \label{inv4}
  \end{align}
  where $b= t \mod (k+1)$ is the index of the current bucket.
\end{definition}

$(p,q) \in C^{k,t}$ means that there is an index $0\le d\le k$ for
which $(p,q)\in B^{k,t}_d$ and $p\in C^{k,t}$ means that there is an
offset $0\le q < 2^l$ for which $(p,q)\in C^{k,t}$.

As the state $t$ is incremented $b$ changes from $0$ to $k$
cyclically.  This can be imagined as a circle turning through $k+1$
positions, justifying its name.  Also, each bucket contains all the
right primes with all the right offsets, that is when it becomes the
current bucket, it will contain exactly those prime-offsets which are
needed for sieving the current segment.

Circles and buckets can be defined for arbitrary $k,t\in\mathbb{N}$,
but only $0\le k \le\lfloor\max P/2^l\rfloor=K$ and $0\le t<n$
are needed. 
%

\begin{theorem}\label{thm:big}
  For each $p\in P$ there is a unique $0\le k\le K$ and for each state
  $t$, a unique $0\le d \le k$ and offset $q$ such that $(p,q) \in
  B^{k,t}_d$.
\end{theorem}
\begin{proof}
  For every $p$ prime, dividing \eqref{inv1} with $2^l$ gives
  $k=\lfloor p/2^l \rfloor$, and if $p \in P$, then $p\le \max P$, so
  $\lfloor p/2^l \rfloor \le \lfloor \max P / 2^l \rfloor=K$,
  therefore $0\le k\le K$.  For each $p$, \eqref{inv1} is true
  independent of the state $t$.
  
  For each $t$, a unique offset $q$ satisfying \eqref{inv2} and a
  unique index $0\le d \le k$ satisfying \eqref{inv3} and \eqref{inv4}
  has to be found . It should be noted that the precondition of
  \eqref{inv3} and \eqref{inv4} are mutually exclusive, that is an
  index $d$ satisfies only one of the two preconditions, and only that
  one has to be proven.

  Medium primes, that is $p<2^l$ is the special case of $k=0$.
  $C^{0,t}$ has only one bucket with the index $d=b=0$, so the
  precondition of \eqref{inv3} is always false.  \eqref{inv2} is
  equivalent to $0\le q < p$ since $p<2^l$.  \eqref{inv4} is
  equivalent to $p\mid 2\bigl((f+t)2^l+q\bigr)+1$ because $0\le b \le
  d$, that is $b=0=d$.  Lemma \ref{lem:q} for the prime $p$ and
  integers $l$ and $t+f$ shows that there is an integer $q$ which
  satisfies \eqref{inv2} and \eqref{inv4}.

  
  For $p>2^l$, the proof is by induction.  If $t=0$ is fixed, then
  $b=0$ is the current bucket's index. The precondition of
  \eqref{inv3} is false and \eqref{inv4} is equivalent to $p\mid
  2\bigl((f+d)2^l+q\bigr)+1$.  Lemma \ref{lem:q} for the prime $p$ and
  integers $l$ and $f$ gives a $q'$ which satisfies $p\mid
  2(f2^l+q')+1$ and $0\le q' < p$. Dividing $q'$ by $2^l$ gives
  $q'=d2^l+q$, where $d$ and $q$ are unique and satisfy \eqref{inv2}
  and \eqref{inv4}.


  If the statement holds for $t\ge 0$, then there is an index $0\le
  d\le k$ and an offset $0\le q < 2^l$ such that $(p,q)\in B^{k,t}_d$.
  The current bucket is $b=t \mod (k+1)$, and the statement will be
  proven for the next state $t'=t+1$ with $b'=(b+1) \mod (k+1)$ as the
  index of the ``next'' current bucket.
  
  The first case is $d \neq b$.  It can be shown that incrementing the
  state, the prime remains in the same bucket with the same offset,
  i.e.\ $(p,q)\in B^{k,t'}_d$. If $b<k$, then $b'=b+1\le k$ holds,
  that is $t'-b'=(t+1)-(b+1)=t-b$ so \eqref{inv3} and \eqref{inv4}
  remain the same, except for the preconditions.  But since $d \neq
  b$, $0\le d < b < b'$ or $b < b' \le d$ will still remain true. If
  $b=k$, then $0\le d < b=k$ and $b'=0=b-k$, so the precondition of
  \eqref{inv3} is false, and \eqref{inv4} becomes:
  \begin{equation*}
    p \mid 
    \bigl( (f+(t+1)+d-(-k))2^l +q \bigr) +1 = 
    \bigl( (f+t+d+k+1)2^l +q \bigr) +1 
  \end{equation*}
  for $0=b' \le d$. Since \eqref{inv3} was true for $d$ and $t$, now
  \eqref{inv4} is true for $d$ and $t+1$.

  The second case is when $d = b$, and it can be shown that
  incrementing the state the prime remains in the same bucket or goes
  into the previous one (modulo $(k+1)$) with a different offset.  If
  $d=b$, $d$ satisfies the precondition of \eqref{inv4}, that is $p
  \mid 2\bigl((f+t)2^l+q\bigr)+1$ is true.  The next odd number
  divisible by $p$ can be obtained by incrementing the offset by $p$.
  As seen in Lemma \ref{lem:kkpo} $q+p$ can be written as
  $q+p=k'2^l+q'$, where $k'=k$ or $k+1$ and $0\le q' < 2^l$.  With
  incrementing the offset by $p$ for $t$ \eqref{inv4} gives:
  \begin{equation}
    p\mid 2\bigl( (f+t+k')2^l+q'\bigr)+1. \label{eq:bj}
  \end{equation}
  Let $d'$ be $d+(k'-1) \mod (k+1)$, that is $d' \equiv d \pmod{k+1}$
  or $d' \equiv d-1 \pmod{k+1}$.  The precondition of \eqref{inv4} for
  the next state $t'$ is true if $b=k$ (then $b'=0$ and $d'=k$ or
  $k-1$) or if $b=0$ and $k'=k$ (then $b'=1$ and $d'=k$). If these
  values are plugged in \eqref{inv4} for $t'$, i.e.\ $p \mid
  \bigl((f+t'+d'-b')2^l+q'\bigr)+1 $, equation \eqref{eq:bj} is
  obtained, which is true. The preconditions of \eqref{inv3} are
  satisfied for every other case, that is, when $0<d=b<k$ (then $0\le
  d'<b'\le k$) or $b=0$ and $k'=k+1$ (then $b'=1$ and $d'=0$). Again,
  by plugging these values in \eqref{inv3} for $t'$, that is $p \mid
  \bigl((f+t'+d'-b'+k+1)2^l+q'\bigr)+1$ equation \eqref{eq:bj} is
  obtained, which is also true.
\end{proof}

As a consequence, if it doesn't cause any confusion, the state may be
omitted from the notation, because each prime with its offset is
maintained only for the current state.  As the program iterates through
states, the primes may ``move'' between buckets, and offsets usually
change.

\medskip

The following Corollary shows, that sieving with circles and buckets
sieves out all composites marked by large primes (sieving with medium
primes is more or less trivial).
\begin{corollary}\label{cor:complete}
  For each $p>2^l$ and odd $i'\in [u,v]$ satisfying $p\mid i'$ there
  exists a unique state $t$ and an offset $q$, so that $(p,q)$ is in
  the current bucket of the circle to which $p$ belongs to.
\end{corollary}
\begin{proof}
  Let $i'$ be represented by $S_i$ for $0\le i < n2^l$, that is
  $i'=2(f2^l+i)+1$.  The statement is true for $t=\lfloor
  i/2^l\rfloor$, because then $i=t2^l+q$, $b = t \mod (k+1)$, and
  substituting $d$ with $b$ in \eqref{inv4}, the equation $p\mid
  2\bigl( (f+t)2^l+q)+1=2(f2^l+i)+1$ is obtained.
\end{proof}

The proof of Theorem \ref{thm:big} could have been simpler, but the
proof by induction gives some insight on how the circles and buckets
work and behave, giving some idea about how to implement them. This
behavior is explicitly stated in the following Corollary.
\begin{corollary}\label{cor:subset}
  For each state $t$, each order $k$, $b=t\mod(k+1)$ and
  $b'=(t+k)\mod(k+1)$, if $(p,q)\in B^{k,t+1}_b$ then $(p,q')\in
  B^{k,t}_b$ for some offset $q'$, and $B^{k,t}_{b'}\subset
  B^{k,t+1}_{b'}$ and for every $b\neq d\neq b'$ and $d\neq d'$
  $B^{k,t}_d=B^{k,t+1}_d$.
\end{corollary}
The first statement says that with respect \emph{only} to primes
$B^{k,t+1}_b \subset B^{k,t}_b$.
\bigskip
\begin{proof}
  The index $b'$ refers to the current bucket in the previous state and as seen in the 
  remarks in the proof of Theorem \ref{thm:big}, iterating from state
  $t$ to $t+1$ leaves the buckets with indexes $d\neq b$ and $d\neq
  b'$ untouched, and some primes with new offsets are left in the
  current bucket while others are put in the previous one.
\end{proof}


\subsection{Modus operandi}
\label{sec:modus-operandi}

The goal is to perform a segment-wise sieve:

Medium primes belonging to $C^{0}$ are a special case, and they sieve
at least once in a segment.  The $t$-th state of $C^0$ contains all
medium primes with the smallest offsets for sieving in the $S^{(t)}$
segment.  For a prime in $C^0$, after sieving with it the offset is
replaced with the smallest offset for sieving the next segment
$S^{(t+1)}$.  After sieving with all medium primes, $C^0$ is in the
$t+1$-th state.  This is implemented in a single loop, iterating
through all medium primes.

The circle $C^k$ ($k>0$), in the $t$-the state, for primes between
$k2^l$ and $(k+1)2^l$, has the prime-offset pairs, needed for sieving
$S^{(t)}$ in the current bucket.  After sieving with all these
primes, the circle is in it's next state, with offsets replaced, and
some primes moved to the previous bucket, ready for sieving
$S^{(t+1)}$.  Sieving large primes is implemented via two embedded
loops, the outer iterating through circles by their order, covering all
primes, and the inner loop iterating through the primes of the current
bucket of the current circle.

The above two procedures are called in a loop for segment $S^{(t)}$,
iterating from $t=0$ to $n-1$.  Corollary \ref{cor:complete} shows,
that the primes for sieving the $t$-th segment are in the current
buckets of circles in $t$-th state, so this procedure performs the
sieve correctly.

\subsection{Implementation}
\label{sec:implementation}

For sequential access, all prime-offset pairs, buckets and circles are
stored as linear arrays: the array of prime-offset pairs is denoted by
$(\hat{p}_i,\hat{q}_i)$, the array of buckets denoted by $\hat{b}_i$
and the array of circles denoted by $\hat{c}_i$ ($i\in\mathbb{N}$).

\begin{figure}[h]
  \centering
  \includegraphics[scale=0.8]{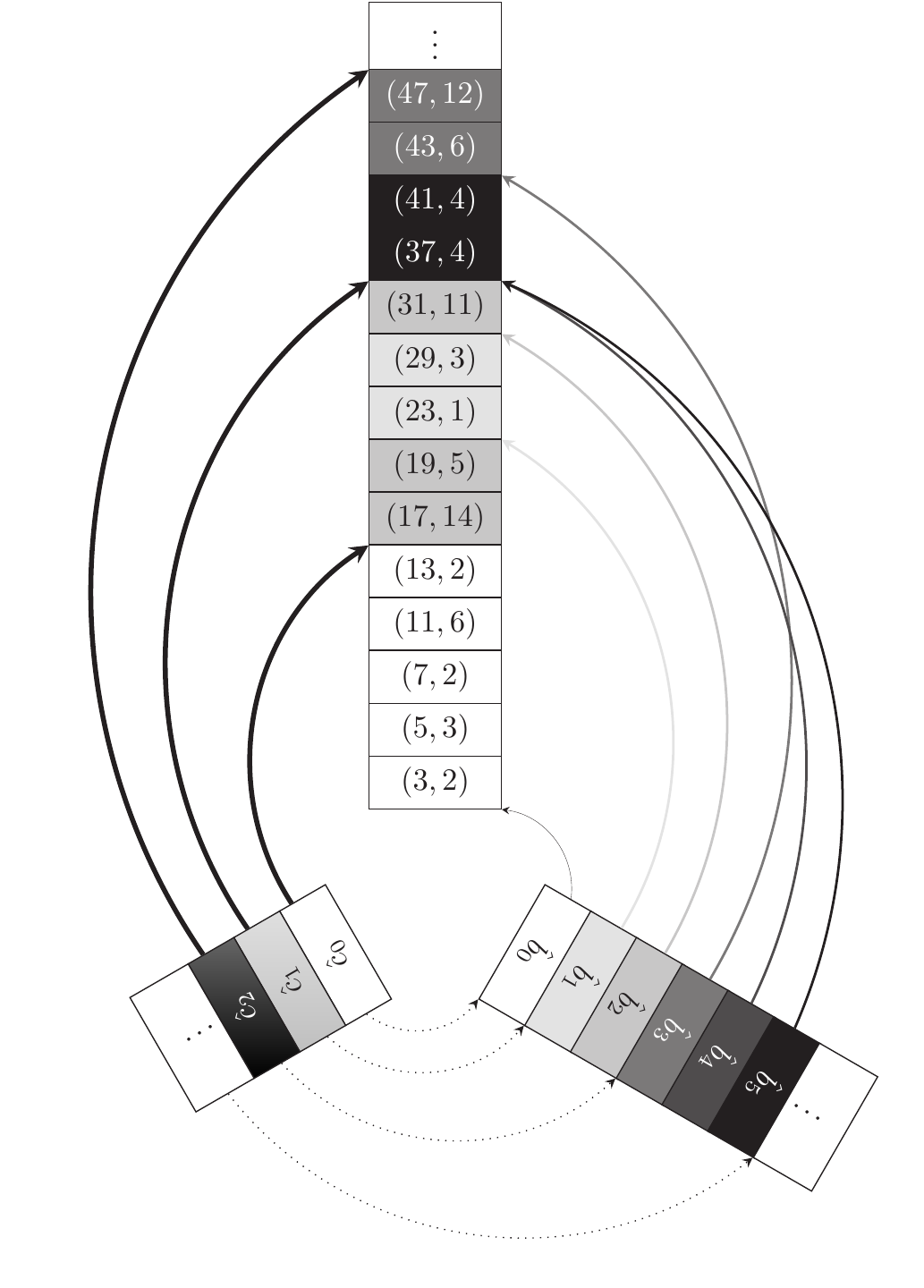}
  \caption{Array of circles, buckets and primes}
  \label{fig:fc}
\end{figure}

\subsubsection{Array of circles}
\label{sec:circles}

$\hat{c}_k$ is the data structure (C \texttt{struct}) implementing the
circle $C^k$.  It is responsible for most of the administration of the
associated primes and buckets.  Of course memory to store $K+1$
circles is allocated.

In the implementation, primes of one circle are a continuous part of
the array of primes.  It was convenient to store the end-pointers of
circles, i.e.\ a pointer to the prime after the last prime in the
circle. So the medium primes are all primes before the up to but not
including the prime at the end-pointer of $\hat{c}_0$, and all primes
in $C^k$ are the primes from the end-pointer of $\hat{c}_{k-1}$ up to
but not including the prime pointed to by the end-pointer of
$\hat{c}_k$.  Since the primes are generated in an ascending order,
these end-pointers can be determined easily, and they don't change
during the execution of the program.

The circle $\hat{c}_k$ maintains the index of the current bucket. It
is incremented by one (modulo $k+1$), that is $b\gets b+1$, if $b<k$
or $b\gets 0$ if $b=k$ assignment is performed after sieving with
primes from the circle. Circle need to maintain the index of the
broken bucket explained in section \ref{sec:broken-bucket}.

The circle $\hat{c}_k$ could also maintain a pointer, to the first
bucket $B^k_0$ represented by $\hat{b}_{k(k+1)/2}$ in the buckets
array, but it is not necessary because it can be calculated from $k$.
There is a more efficient solution if the circles are processed with
ascending orders: The starting bucket of $\hat{c}_1$ is $\hat{b}_1$
and this is stored as a temporary pointer. For every circle $\hat{c}_{k+1}$ the 
starting bucket is at $k+1$ buckets after the first bucket of the
previous circle $\hat{c}_k$, so when finished with circle $\hat{c}_k$,
this pointer has to be increased by $k+1$.

In Figure \ref{fig:fc} the value $l$ is $4$ so the cache size is $16$. $C^0$ has a white 
background as well as the bucket and primes associated with it. $C^1$
is light gray with black text, with two different shades for the two buckets and primes in
them, and in a similar way $C^2$ is black with white text and slightly lighter shades of gray for
the buckets and primes. The end-pointers are drawn as thick arrows,
indicating that they don't move during the execution.  The dotted
lines are the calculable pointers to the first buckets.

\subsubsection{Array of buckets}
\label{sec:buckets}

All buckets are stored consecutively in one array, i.e.\ the bucket
$B^k_d$ of circle $C^k$ for $0\le d \le k$ is represented by $\hat{b}
_{k(k+1)/2+d}$. There are $K+1$ circles, with $k+1$ buckets for each
$0\le k \le K$, so memory for storing $(K+1)(K+2)/2$ buckets needs to
be allocated.

The value of each $\hat{b}_d$ is the index
(\texttt{uint32\textunderscore{}t}) of the first prime-offset pair
which belongs to the bucket $\hat{b}_d$.  Primes that belong to one
bucket are also in a continuous part of the primes array, so the bucket
$B^k_d$ contains the primes $\hat{p}_i$ (and the associated offsets
$\hat{q}_i$) for $\hat{b}_{d'}\le i < \hat{b}_{d''} $ where
$d'=k(k+1)/2+d$ and $d''=k(k+1)/2 + (d+1)\mod(k+1)$.  The  broken
bucket is an exception to this.  Empty buckets are represented with
entries in the buckets array having the same value, i.e.\ $B^k_d$ is
empty if $\hat{b}_{d'}=\hat{b}_{d''}$ for $d'= k(k+1)/2+d$ and
$d''=k(k+1)/2+(d+1)\mod(k+1)$, e.g.\ $\hat{b}_4$ is empty in Figure
\ref{fig:fc}.

Buckets are set during the initialization, but change constantly
during sieving. First, using Lemma \ref{lem:q} an offset $0\le q' < p$
is found for each $p$ prime. All prime-offset pairs of the circle
$C^k$ are sorted in ascending offsets.  Similarly, as primes are
collected in circles, within one circles the offsets are collected in
buckets.  $B^k_d$ contains all prime-offset pairs, so that $d2^l\le q'
< (d+1) 2^l$, but instead $q'$, $0\le q=q'-d2^l<2^l$ is stored.

For each circle $C^k$, for each pair $(p,q)\in B^k_b$, the bit with
index $q$ in the current segment is set.  After that, the new offset
$q' = q+p$ is calculated, which is, because of Lemma \ref{lem:kkpo},
either $k2^l\le q' < (k+1)2^l$ or $(k+1)2^l\le q' < (k+2)2^l$.  In the
former case $q'-k2^l$ is stored in the previous bucket (modulo $k+1$),
or in the latter case $q'-(k+1)2^l$ is stored in the current bucket,
as described in Corollary \ref{cor:subset}.

This can be implemented efficiently by keeping copies of a prime from the 
beginning and another prime from the end of the bucket.  That way, in
the first case ($k$ segments were skipped), the prime-offset pair in
the beginning of the bucket is overwritten and the value in the
buckets array indicating the beginning of the bucket is incremented,
thus putting the replaced, new prime-offset pair, in the previous
bucket.  In the second case ($k+1$ segments were skipped), the
prime-offset pair in the end of the bucket is overwritten with the new
prime-offset pair and it stays in the current bucket.  After replacing
a pair in one of the ends (the beginning or the end) of the bucket,
the next prime is read from that end of the bucket, that is from the
next entry, closer to the center of the bucket.

\subsubsection{Array of primes}
\label{sec:primes}

The array of primes contains all primes (medium and large), with the
appropriate offsets, needed for sieving. The prime-offset pairs are
stored as two 32 bit unsigned integers (two
\texttt{uint32\textunderscore{}t}s in a \texttt{struct}). Enough
memory to store about $\frac{\sqrt{v}}{\log\sqrt{v}}$ pairs is
allocated.

The array of primes is filled during the initialization phase. The
values of the offsets $q$ change after finishing a segment.  Sometimes
pairs from the end and the beginning of a bucket are swapped (as
explained earlier), but this is all done in-place, that is, the array
itself does not need to be modified or copied, just the values stored.
All primes that belong to one circle as well as those that belong to
one bucket (except the broken bucket) are stored in a coherent and
continuous region of memory.

\subsubsection{Broken bucket}
\label{sec:broken-bucket}

For the circle $C^k$, the index of the broken bucket is $r=\max\{d:
\hat{b}_{k(k+1)/2+d} = M\}$, where $M = \max \{\hat{b}_d: {k(k+1)}/{2}
\le d < {(k+1)(k+2)}/{2}\}$.  The primes which belong to this bucket,
are the ones from the index $\hat{b}_{k(k+1)/2+r}$ and up to but not
including the prime at the end-pointer of $\hat{c}_k$ and the primes
from the end-pointer of $\hat{c}_{k-1}$ up to but not including the
prime with the index $\hat{b}_{k(k+1)/2+r'}$ where
$r'=(r+1)\mod(k+1)$, e.g.\ $\hat{b}_2$ in Figure \ref{fig:fc}.  This
idea also justifies the name circles, because logically the next prime
after the end-pointer of a circle is the first prime of the circle.
When the broken bucket is not actually broken, the value of
$\hat{b}_{k(k+1)/2+r'}$ is set to the index of the prime at the
end-pointer of $\hat{c}_{k-1}$, e.g.\ $\hat{b}_3$ in Figure
\ref{fig:fc}.

Every circle has a broken bucket and this has to be stored as a
variable for each circle.  The fact, that this can not be omitted is
not trivial, but if all primes of a circle are one bucket, then all
other buckets in that circle are empty.  Because empty buckets are
represented by having the same value as the following bucket, all
buckets in that circle, that is all entries of $\hat{b}_d$ which
represent the buckets of that circle, will have the same value. In
this situation the program can't decide which buckets are empty and
which one contains all the primes.

The broken bucket also moves around.  When sieving with a bucket, its
lower boundary is incremented.  If sieving with the broken bucket,
when the beginning of the bucket moves past the end of the circle (and
jumps to the beginning), the previous bucket (modulo $k+1$) becomes
the new broken bucket.

\section{Speeding up the algorithm}
\label{sec:speed-up-algo}

The roughly described implementation of sieving can be further refined
to gain valuable performance boosts.

\subsection{Small primes}
\label{sec:small-primes}

Sieving with primes $p<64$ can be sped up by not marking individual
bits, but rather applying bit masks. The subset of medium primes below
$64$ are called \emph{small} primes.

The \emph{AMD64}\footnote{AMD\texttrademark \ is a trademark of Advanced
  Micro Devices, Inc.} architecture processors with \emph{SSE2}
extension, have sixteen 64-bit general purpose R registers, and sixteen
128-bit XMM registers. For sieving with small primes, the generated
64bit wide bit masks are loaded in these registers and \texttt{or}-ed
together, to form the sieve table with small primes applied to it. The
masks are then \texttt{shift}-ed, to be applied to the next 64 bits
for R registers and 128 bits for XMM registers.

This is of course done in parallel, sieving by 128 bits at a time. The
XMM registers first 64 bits are loaded from the memory at the
beginning of sieving of a segment, and the last 64 bits are
\texttt{shift}-ed (just like the R registers are shifted ``mod
64''). There is two times as much sieving with the R registers than
with the XMM registers.

With the first four primes ``merged'' into two, all the small primes
can fit in the R and XMM registers, so the only memory access is
sequential and done once when starting and once when finished
sieving. The primes $3$ and $11$ are merged into $33$, that is, the
masks of $3$ and $11$ are combined at the initialization, and the
shifting needs to be done as if $33$ was the prime for sieving,
because the pattern repeats after $33$ bits. $5$ and $7$ are merged
into $35$ and treated similarly.

\subsection{Medium primes}
\label{sec:medium-primes}

As described earlier, for $(p,q)$ pairs with medium primes, sieving
starts from $q$ by increasing it by $p$ after sieving, until $q\ge
2^l$.  Then the sieving is finished for that segment, and the sieving
of the next segment starts from $q'=q-2^l$. There are two methods in
which this algorithm can be sped up.

\subsubsection{Wheel sieve}
\label{sec:wheel-sieve}

In the special case of the sieve of Eratosthenes, the ``wheel''
algorithm (described in \cite{wheel}) can be used to speed up the
program.  In some sense, it is an extension of the idea of not sieving
with number $2$.

Let $W$ be the set of the first few primes and $w=\prod_{p\in W}p$.
Sieving with the primes from $W$, sieves out a major part of the sieve
table, and these bits can be skipped.  Basically, when sieving with a
prime $p\not\in W$, the number $i$ needs to be sieved (marked) by $p$,
only if it is relative prime to $w$, that is, if $i$ is in the reduced
residue system modulo $w$ denoted by $W'$ (if $i\not\in W'$ some prime
from $W$ will mark it).  

Let $w_0<\cdots<w_{\varphi(w)-1}$ be the elements of $W'$, and
$\Delta_s$ the number of bits that should be skipped, after sieving
the bit with index congruent to $w_s$, that is $\Delta_i= (w +
w_{(i+1) \mod \varphi(w)} - w_i) \mod w$.  When $i$ is sieved out by
$p\not\in W$, instead of sieving $i+p$ next, the program can skip to
$i+\Delta_s p$ if $i\equiv w_s\pmod{w}$.

In the implementation, $W=\{2,3,5\}$, but $2$ is ``built in'' the
representation and this complicates thing a little bit: $w=15$,
$\varphi(w)=8$ and $w_0=0$, $w_1=3$, $w_2=5$, $w_3=6$, $w_4=8$,
$w_5=9$, $w_6=11$, $w_7=14$ are used (instead of $w=30$ and $1$, $7$,
$11$, $13$, $17$, $19$, $23$, $29$ for $w_s$).  For each prime the
offset $q$ is initialized to the value $q'+mp$, where $q'$ the offset
found using Lemma \ref{lem:q} and $m$ is the smallest non-negative
integer, so that $f2^l + q \equiv w_s \pmod{w}$ for some $0\le s
\le7$.

Let $p^{-1}$ be the inverse of $p$ modulo $15$, and $x$ a non-negative
integer so that:
\begin{equation}
f2^l+q + xp \equiv 7 \pmod{15}.\label{eq:x}
\end{equation}
Note that the residue class represented by $7$ is $15$, and that is
the class divisible both by $3$ and $5$, and it is in a sense the
``beginning'' of the pattern generated by the primes in $W$ when
sieving.  \eqref{eq:x} states that after $x$ times sieving (regularly)
with $p$, the offset is at the ``beginning'' of the pattern, that is,
in the residue class represented by $7$, so $y=7-x$ is the residue
class in which $q$ actually is.  $x\equiv(7-(f2^l+q))p^{-1} \pmod{15}$
can be obtained from \eqref{eq:x} by multiplying it with $p^{-1}$.

There is an index $0\le s \le 7$, so that $w_s=y$.  The index $s$,
indicating where in the pattern is the offset $q$, is stored beside
each $(p,q)$ pair. Before sieving with $p$, the array
$\Delta_0p,$ \ldots, $\Delta_7p$ is generated in memory, and a pointer is
set to $\Delta_sp$. After marking a bit, the offset is incremented by
the values found at that pointer, and the pointer is incremented
modulo $8$, which can be implemented very efficiently with a logical
\texttt{and} operation and a bit mask.  Also all prime-offset pairs
are stored on 64 bits and medium primes are $p<2^l$ (where $l$ is
never more than $30$), so at least 4 bits are not used where the index
$0\le s\le 7$ can fit.

\subsubsection{Branch misses}
\label{sec:branch-misses}

Another speed boost can be obtained by treating the \emph{larger
  medium primes} (near to $2^l$) differently.  This idea is somewhat
similar to the one used with circles, because it is based on the
observation that, if $\frac{2^l}{(k+1)} < p < \frac{2^l}{k}$, then $p$
sieves $k$ or $k+1$ times in one segment ($0<k\in\mathbb{N}$).  There
is a different procedure $g_k$, for each of the first few values of
$k$ (e.g.\ $0<k<16$). $g_k$ iterates the offset $k+1$ time, with the
last iteration implemented using conditional move (\texttt{cmov})
operations.  So, for each $k$, primes $\frac{2^l}{(k+1)}< p <
\frac{2^l}{k}$ are collected in a different array, and the procedure
$g_k$ is invoked for each prime in that array.  Having fixed number of
iterations with a conditional move is faster then a branch miss,
because the CPUs instruction stream is not interrupted.

\subsection{Large primes}

The sieving with large primes is roughly described above.  Sieving
with one prime, putting it back, with the new offset, and modifying
the bucket boundary can be accomplished with only about 15 assembly
instructions using conditional moves (\texttt{cmov}). This is very
efficient, but other techniques can also be applied to reduce
execution time.

\subsubsection{Interleaved processing}
\label{sec:interl-proc}

Because the order in which the primes are processed doesn't matter,
the memory latency can be hidden by processing primes from both ends
of the bucket.  As described earlier, for each bucket, two
prime-offset pairs are loaded from the beginning and end of the
current bucket and one of them is processed.  To hide memory latency,
the next prime is loaded into place of the processed prime
\emph{while} the other one is being processed.  ``Processing a prime''
covers the following steps: marking the bit at the offset $q$;
determining if $q+p$ skips $k$ or $k+1$ segments; calculating the new
offset $q'\gets q+p-k2^l$, replacing the pair at the beginning of the
bucket and incrementing the bucket's lower boundary, for the former
case; or in the latter case, decrementing the pointer indicating the
finished primes at the top of the bucket, after replacing the pair at
the end of the bucket with the offset $q'\gets q+p-(k+1)2^l$.
Processing of one prime is about 15-18 assembly instructions, which is
approximately 5-6 clock cycles on today's processors, about the same
time needed for the other prime to be loaded in the registers.

\vspace*{-0.1cm}
\subsubsection{Broken bucket and loop unrolling}
\label{sec:broken-bucket-loop}

The well-known technique of loop unrolling can efficiently be used for
processing primes-offset pairs.  The core of the loop described above,
which processes two primes terminates when the difference between the
pointer from the beginning and end of the bucket becomes zero.  With
right \texttt{shift} and a logical \texttt{and} instructions, the
quotient $a$ and remainder $r$ of this difference when divided by
$2^h$ can be obtained (e.g.\ $h=4$ or $5$).  Then the loop core can
be executed $a$ times in batches of $2^h$ runs, and afterward $r$
times, thus reducing the time spent on checking if the difference is
zero.

The loop unrolling of the broken bucket is a bit trickier, but
manageable.  Let $\delta_1$ denote the difference between the
beginning of the bucket and the end of the circle, and $\delta_2$ the
difference between the beginning of the circle and the end of the
bucket.  The difference used for unrolling, as described above, would
be $\delta_1+\delta_2$ but the when modifying the pointers after
processing a prime, it would have to be checked, if it moves past the
beginning or end of the circle (to jump to the other side).  Instead,
the unrolling is applied to $\min\{ \delta_1, \delta_2\}$.  Since it
can't be predicted if the beginning or end pointer is going to be
modified, the values of $\delta_1$ and $\delta_2$, the maximum,
quotient $a$ and reminder $r$ have to be reevaluated after each batch,
until one of the pointers ``jump'' to the other side.  Then the bucket
will no longer be broken, so the simpler unrolling described above can
be applied.

\section{Results}
\label{sec:results}
The program was run on (a single core of) two computers referred to by
their names \emph{lime} and \emph{complab07}.  The goal was to
supersede the implementation found in the speed comparison chart of
\cite{tos}, but the results can not be compared directly, because of
the differences in hardware.  Our implementation, running on
\emph{lime} would come in 7th and \emph{complab07} the 16th in the
speed comparison chart, but with significantly slower memory.

\begin{table}[htb]
  \begin{center}
    \begin{tabular}{|l|l|l|l|l|l|l|}
\hline
$e$ & lime & cl07 & [a0F80] & [a0FF0] & [i06E8] & [a0662] \\ \hline
1e12 & 1.45 & 2.09 & 0.57 & 0.68 & 1.28 & 1.07 \\ \hline
2e12 & 1.45 & 2.10 & 0.64 & 0.75 & 1.37 & 1.17 \\ \hline
5e12 & 1.45 & 2.27 & 0.74 & 0.85 & 1.48 & 1.29 \\ \hline
1e13 & 1.45 & 2.28 & 0.80 & 0.92 & 1.57 & 1.38 \\ \hline
2e13 & 1.46 & 2.38 & 0.86 & 0.99 & 1.66 & 1.47 \\ \hline
5e13 & 1.45 & 2.46 & 0.95 & 1.08 & 1.76 & 1.59 \\ \hline
1e14 & 1.46 & 2.5 & 1.01 & 1.14 & 1.85 & 1.67 \\ \hline
2e14 & 1.45 & 2.58 & 1.08 & 1.21 & 1.94 & 1.76 \\ \hline
5e14 & 1.68 & 2.70 & 1.16 & 1.29 & 2.04 & 1.87 \\ \hline
1e15 & 1.63 & 2.80 & 1.22 & 1.36 & 2.11 & 1.96 \\ \hline
2e15 & 1.71 & 2.86 & 1.28 & 1.42 & 2.19 & 2.06 \\ \hline
5e15 & 1.83 & 2.95 & 1.37 & 1.50 & 2.29 & 2.20 \\ \hline
1e16 & 1.89 & 3.04 & 1.42 & 1.56 & 2.37 & 2.32 \\ \hline
2e16 & 1.95 & 3.13 & 1.49 & 1.63 & 2.45 & 2.47 \\ \hline
5e16 & 2.03 & 3.25 & 1.58 & 1.75 & 2.57 & 2.72 \\ \hline
1e17 & 2.08 & 3.34 & 1.64 & 1.86 & 2.67 & 2.93 \\ \hline
2e17 & 2.15 & 3.45 & 1.72 & 2.02 & 2.79 & 3.23 \\ \hline
5e17 & 2.22 & 3.60 & 1.84 & 2.21 & 2.96 & 3.66 \\ \hline
1e18 & 2.29 & 3.75 & 1.99 & 2.39 & 3.13 & 4.03 \\ \hline
2e18 & 2.33 & \multicolumn{1}{r|}{0} & 2.26 & 2.61 & 3.31 & 4.52 \\ \hline
\end{tabular}


  \end{center}
  \caption{Execution times in seconds for intervals of 
    $10^9\approx 2^{30}$ with the midpoint at $10^e$\label{tab:speed}}
\end{table}

\begin{table}[htb]
  \begin{center}
    {\small
\begin{tabular}{|l|l|}
  \hline
  lime  & 2000MHz Intel Core2 Duo (E8200) 
  model 23, stepping 6, DDR2 666MHz\!\!\!\\ \hline
  cl07 & 1595MHz AMD Athlon64 3500+,
  model 47, stepping 2, DDR 200MHz \\ \hline
  a0F80 & 2600MHz 6-Core AMD Opteron (Istanbul), 
  model 8, stepping 0, DDR2 \\ \hline
  a0FF0 & 2210MHz Athlon64 (Winchester), 
  model 15, stepping 0, DDR 333\\ \hline
  i06E8 & 1830MHz T2400 (Core Duo),
  model 14, stepping 8, DDR2 533\\ \hline
  a0662 & 1669MHz Athlon (Palomino), 
  model 6, stepping 2, DDR 333\\ \hline
\end{tabular}
}

  \end{center}
  \caption{The CPU and memory configurations of the computers 
    used for measurements   \label{tab:specs}}
\end{table}

\begin{figure}[htb]
  \centering
  \includegraphics[scale=1.5]{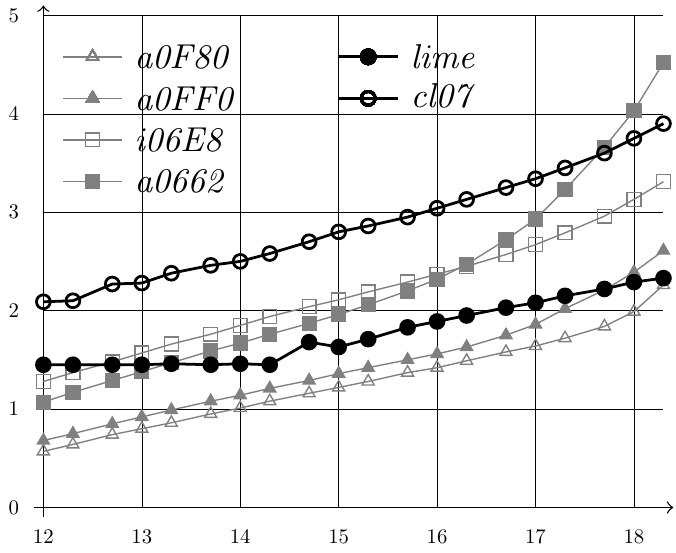}
  \caption{Speed comparison chart}
  \label{fig:plot}
\end{figure}

Compared to the 666MHz DDR2 memory of \emph{lime}, the first five or
so computers from the speed comparison chart have memory speeds of
800MHz and above, and with a better memory our implementation could
probably compete better.  But the real improvement can be seen, when
running on older hardware, like \emph{complab07}.  With the slow
memory of 200MHz, the plot in Figure \ref{fig:plot}, is much
flatter and closer to the theoretical speed of $n \log\log n$ than for
example the similar \emph{i0662} with a faster 333MHz RAM.

It should also be noted, that the major part of execution is spent on
sieving with medium primes and more optimization is desired out of
that part of the algorithm. We also had some unexpected difficulties
optimizing assembly code for the Intel processors, due to confusing
documentation and slow execution of the \texttt{bts} (bit test set)
instruction.

For \emph{lime}, \emph{complab07} and some computers from \cite{tos},
Table \ref{tab:speed} shows the time needed to sieve out an interval
(represented by $2^{30}\approx 10^{9}$ bits, with its midpoint at
$10^e$).  This data is plotted out in Figure \ref{fig:plot}: values of
$e$ are represented on the horizontal, execution times in seconds on
the vertical axis.

\section{Future work}
\label{sec:future-work}

The program was originally written for verifying the Goldbach
conjecture, but only the sieve for generating the table of primes was
finished and measured because that takes up the majority of the work
for the verification.  The completion of the verification application
would be desirable.  Also the current implementation supports sieving with primes 
only up to 32 bits, on current architectures, the implementation of
sieving with primes up to 64 bits would not be a problem.

Most of the techniques described here (except the wheel algorithm),
especially the use of circles and buckets can be applied for a wider
range of sieving algorithms.  For example, in \cite{gnfs} a similar
attempt is made to exploit the cache hierarchy, but the behavior of
the large primes is more predictable with our method, and even an
implementation for processors not designed for sieving algorithms is 
possible.  Therefore the multiple polynomial quadratic sieve, on the
Cell Broadband Engine Architecture\footnote{Cell Broadband
  Engine\texttrademark \ is a trademark of Sony Computer Entertainment
  Incorporated}, with 128K byte cache (i.e.\ Local Store) controlled
by the user via DMA, can be implemented efficiently.  Further
performance can be gained by combining buckets and circles with
parallel processing: sieving with different polynomials on different
processors for MPQS-like algorithms, and a segment-wise pipeline-like
processing for algorithms similar to the sieve of Eratosthenes.









\section*{Acknowledgements}


The Project is supported by the European Union and co-financed by the
European Social Fund (grant agreement no. T\'AMOP
4.2.1/B-09/1/KMR-2010-0003).




\bigskip
\rightline{\emph{Received: October 2, 2011 {\tiny \raisebox{2pt}{$\bullet$\!}} Revised: November 8, 2011}} 

\end{document}